\documentclass[pra,11pt,nofootinbib]{revtex4-2}

\usepackage{verbatim}
\usepackage[utf8]{inputenc}
\usepackage[T1]{fontenc}
\usepackage{amsmath,amssymb,amsfonts,mathtools,amsthm}
\usepackage{graphicx}
\usepackage{xcolor}
\usepackage{bbm}
\usepackage[colorlinks=true,linkcolor=blue,citecolor=blue,urlcolor=blue]{hyperref}
\usepackage{tikz}
\usetikzlibrary{arrows.meta,calc,positioning,shapes.geometric}
\newtheorem{thm}{Theorem}[section]
\newtheorem{prop}[thm]{Proposition}

\newtheorem{lemma}[thm]{Lemma}
\newtheorem{conj}[thm]{Conjecture}

\newtheorem{rem}[thm]{Remark}

\renewcommand{\H}{\mathcal{H}}

\newcommand{\E}{\mathbb{E}}
\newcommand{\Id}{\mathrm{Id}}
\newcommand{\Tr}{\mathrm{Tr}}

\newcommand{\ket}[1]{\lvert #1\rangle}
\newcommand{\bra}[1]{\langle #1\rvert}

\newcommand{\ketbra}[2]{\lvert #1\rangle\!\langle #2\rvert}

\begin{document}

\title{\textbf{Exponential lower bounds for low-degree strategies \\ in position-based quantum cryptography}}
\author{I.M. Moreno-Cuadrado}

\author{C. Palazuelos}
\affiliation{\small Departamento de An\'alisis Matem\'atico y Matem\'atica Aplicada, Universidad Complutense de Madrid, 28040 Madrid, Spain}
\affiliation{\small Instituto de Ciencias Matematicas, 28049 Madrid, Spain}

\author{D. Perez-Garcia}
\affiliation{\small Departamento de An\'alisis Matem\'atico y Matem\'atica Aplicada, Universidad Complutense de Madrid, 28040 Madrid, Spain}
\affiliation{\small Instituto de Ciencias Matematicas, 28049 Madrid, Spain}

\begin{abstract}

Non-local quantum computation (NLQC) consists on the implementation of a bipartite unitary $U_{AB}$ by two cooperating distant players by means of a two-round protocol with a single intermediate round of simultaneous communication. It is a major open question to know whether there exists a unitary $U_{AB}$ whose implementation in NLQC requires extra quantum systems with dimensions ${\rm exp}(n^{\Omega(1)})$, where $n$ is the dimension of systems $A$ and $B$. This type of exponential lower bound is essential for security of position-based quantum cryptography, since the capabilities of the adversaries are precisely described by NLQC. It has also been connected to many other problems, such as the optimality of universal quantum simulators, or the expected properties of holographic quantum gravity. Currently the best lower bounds are only sublinear in $n$.

In this paper we consider the family of diagonal ${\pm 1}$ valued unitaries $U_{\varepsilon}$, indexed by an element of the Boolean hypercube $\varepsilon \in \{\pm 1\}^{n^2}$, i.e. $U_{\varepsilon}|i\rangle_A|j\rangle_B =\varepsilon_{ij} |i\rangle_A|j\rangle_B$. We show that if $U_{\varepsilon}$ can be implemented in NLQC with constant accuracy for all $\varepsilon$, and  the dependency on $\varepsilon$ in the first round of the associated strategy is a polynomial of degree $O(n^{\frac{1}{4}-\delta})$, then the implementation of at least one $U_{\varepsilon}$ in NLQC requires resources scaling as $\exp\left(\Omega\left(\frac{n^{2\delta}}{(\log n)^2}\right)\right)$. The proof is based on geometric properties of particular Banach spaces, namely their type constants, together with random estimates of Boolean functions valued on them. 

\end{abstract}

\maketitle

\section{Introduction}
Quantum position verification (QPV) aims to provide secure authentication based on the geographical position of the communicating parties. This is expected to be relevant in settings where location itself is the identifying property, such as the navigation of autonomous cars or drones. First proof-of-principle experiments in these contexts have already been reported \cite{KavuriEtAl2025,KavuriEtAl2026}.

The most basic setting is one-dimensional QPV. Two cooperating verifiers send quantum systems $A$ and $B$ to a prover, timed so that they arrive simultaneously at the position to be verified. The honest prover, located at that position, is expected to implement a global unitary $U$ on the joint system $AB$ and return systems $A$ and $B$ to the corresponding verifiers. We will work in the idealized regime in which signals travel at the speed of light and the time required to implement $U$ is negligible compared with the time needed for the systems to travel between the verifiers and the prover.

From the viewpoint of two coordinated adversaries placed on opposite sides of the honest position, breaking one-dimensional QPV is equivalent to succeeding at a task of non-local quantum computation (NLQC). The adversaries must implement the global unitary $U$ by means of a two-round protocol with a single round of simultaneous communication; see Figure~\ref{fig:causal-structure}.

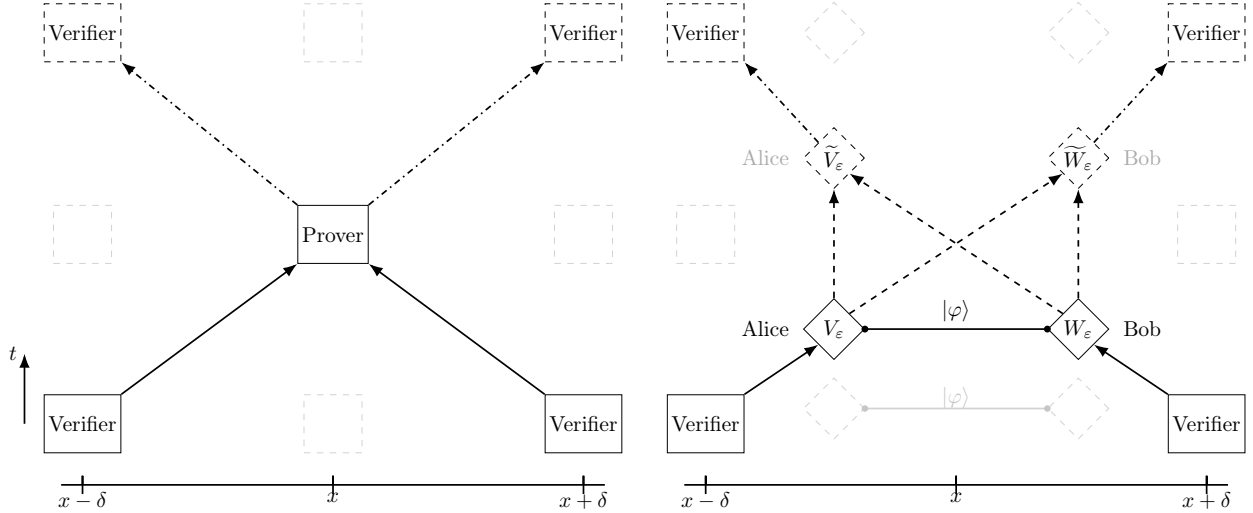
\begin{figure}[t]
\centering
\resizebox{\textwidth}{!}{%
\begin{tikzpicture}[
    >=Latex,
    verifier/.style={
        draw,
        rectangle,
        minimum width=.95cm,
        minimum height=.95cm,
        inner sep=2pt,
        font=\small
    },
    verifierdash/.style={
        draw=black,
        dashed,
        rectangle,
        minimum width=.95cm,
        minimum height=.95cm,
        inner sep=2pt,
        font=\small
    },
    prover/.style={
        draw,
        rectangle,
        minimum width=.95cm,
        minimum height=.95cm,
        inner sep=2pt,
        font=\small
    },
    op/.style={
        draw,
        diamond,
        aspect=1.0,
        minimum width=1.00cm,
        minimum height=1.00cm,
        inner sep=0pt,
        font=\small
    },
        opdash/.style={
        draw=black,
        dashed,
        diamond,
        aspect=1.0,
        minimum width=1.00cm,
        minimum height=1.00cm,
        inner sep=0pt,
        font=\small
    },
    ghost/.style={
        draw=gray!35,
        text=gray!45,
        dashed,
        rectangle,
        minimum width=.95cm,
        minimum height=.95cm,
        inner sep=2pt,
        font=\small
    },
    ghostop/.style={
        draw=gray!35,
        text=gray!45,
        dashed,
        diamond,
        aspect=1.0,
        minimum width=1.00cm,
        minimum height=1.00cm,
        inner sep=0pt,
        font=\small
    },
    qarrow/.style={->, thick},
    carrow/.style={->, thick, dash dot},
    comm/.style={->, thick, dashed},
    ghostline/.style={gray!35, thick},
    baseline/.style={thick},
    namelabel/.style={text=gray!65, font=\small},
    toplabel/.style={text=gray!65, font=\small}
]

\begin{scope}[xshift=0cm]

\coordinate (L0) at (0,0);
\coordinate (M0) at (4.1,0);
\coordinate (R0) at (8.2,0);

\coordinate (VLlow) at (0,1.0);
\coordinate (VRlow) at (8.2,1.0);
\coordinate (P)     at (4.1,4.1);
\coordinate (VLup)  at (0,7.4);
\coordinate (VRup)  at (8.2,7.4);

\draw[baseline] (-.35,0) -- (8.55,0);
\draw[baseline] (L0) -- ++(0,.15) -- ++(0,-.3);
\draw[baseline] (M0) -- ++(0,.15) -- ++(0,-.3);
\draw[baseline] (R0) -- ++(0,.15) -- ++(0,-.3);

\node[below] at (L0) {$x-\delta$};
\node[below] at (M0) {$x$};
\node[below] at (R0) {$x+\delta$};

\draw[->, thick] (-.95,1.0) -- (-.95,2.15);
\node[left] at (-.95,2.15) {$t$};

\node[verifier] (vll) at (VLlow) {Verifier};
\node[verifier] (vrl) at (VRlow) {Verifier};
\node[prover]   (pr)  at (P)     {Prover};

\node[verifierdash] (vlu) at (VLup) {Verifier};
\node[verifierdash] (vru) at (VRup) {Verifier};

\node[ghost] at (4.1,1.0) {};
\node[ghost] at (0,4.1) {};
\node[ghost] at (8.2,4.1) {};
\node[ghost] at (4.1,7.4) {};

\draw[qarrow] (vll.north east) -- (pr.south west);
\draw[qarrow] (vrl.north west) -- (pr.south east);

\draw[carrow] (pr.north west) -- (vlu.south east);
\draw[carrow] (pr.north east) -- (vru.south west);

\end{scope}

\begin{scope}[xshift=10.2cm]

\coordinate (L0) at (0,0);
\coordinate (M0) at (4.1,0);
\coordinate (R0) at (8.2,0);

\coordinate (VLlow) at (0,1.0);
\coordinate (VRlow) at (8.2,1.0);

\coordinate (A1) at (2.1,2.55);
\coordinate (B1) at (6.1,2.55);

\coordinate (A2) at (2.1,5.35);
\coordinate (B2) at (6.1,5.35);

\coordinate (VLup) at (0,7.4);
\coordinate (VRup) at (8.2,7.4);

\draw[baseline] (-.35,0) -- (8.55,0);
\draw[baseline] (L0) -- ++(0,.15) -- ++(0,-.3);
\draw[baseline] (M0) -- ++(0,.15) -- ++(0,-.3);
\draw[baseline] (R0) -- ++(0,.15) -- ++(0,-.3);

\node[below] at (L0) {$x-\delta$};
\node[below] at (M0) {$x$};
\node[below] at (R0) {$x+\delta$};

\node[verifier] (vll) at (VLlow) {Verifier};
\node[verifier] (vrl) at (VRlow) {Verifier};

\node[op] (Aop) at (A1) {$V_\varepsilon$};
\node[op] (Bop) at (B1) {$W_\varepsilon$};

\node[opdash] (Atop) at (A2) {$\widetilde V_\varepsilon$};
\node[opdash] (Btop) at (B2) {$\widetilde W_\varepsilon$};

\node[verifierdash] (vlu) at (VLup) {Verifier};
\node[verifierdash] (vru) at (VRup) {Verifier};

\node[left=3pt of Aop] {Alice};
\node[right=3pt of Bop] {Bob};
\node[toplabel, left=3pt of Atop] {Alice};
\node[toplabel, right=3pt of Btop] {Bob};

\draw[thick] (Aop.east) -- (Bop.west);
\node[above] at ($(Aop.east)!0.5!(Bop.west)$) {$|\varphi\rangle$};

\fill (Aop.east) circle (1.6pt);
\fill (Bop.west) circle (1.6pt);

\draw[qarrow] (vll.north east) -- (Aop.south west);
\draw[qarrow] (vrl.north west) -- (Bop.south east);

\draw[comm] (Aop.north) -- (Atop.south);
\draw[comm] (Bop.north) -- (Btop.south);

\draw[comm] (Aop.north east) -- (Btop.south west);
\draw[comm] (Bop.north west) -- (Atop.south east);

\draw[carrow] (Atop.north west) -- (vlu.south east);
\draw[carrow] (Btop.north east) -- (vru.south west);

\node[ghostop] (Ag) at (2.1,1.25) {};
\node[ghostop] (Bg) at (6.1,1.25) {};
\draw[ghostline] (Ag.east) -- (Bg.west);
\node[gray!45] at ($(Ag.east)!0.5!(Bg.west)+(0,.18)$) {$|\varphi\rangle$};
\fill[gray!45] (Ag.east) circle (1.5pt);
\fill[gray!45] (Bg.west) circle (1.5pt);

\node[ghost] at (0,4.1) {};
\node[ghost] at (8.2,4.1) {};
\node[ghostop] at (2.1,7.4) {};
\node[ghostop] at (6.1,7.4) {};

\end{scope}

\end{tikzpicture}%
}
\caption{Causal structure of one-dimensional PV protocols. Honest implementation (left) vs.\ adversarial scenario (right).}
\label{fig:causal-structure}
\end{figure}

The classical analogue is information-theoretically insecure. Since classical information can be copied and forwarded, the adversaries can simply copy the received challenges, exchange the relevant information in the first round, and emulate the honest prover in the second. Thus classical position verification is impossible without additional assumptions \cite{ChandranEtAl2009,BuhrmanEtAl2014}. Computational assumptions can change this conclusion: in particular, protocols with classical verifiers have been constructed under assumptions such as the quantum hardness of Learning with Errors \cite{LiuLiuQian2022}.

NLQC has since become a common language connecting QPV with several other topics:  port-based teleportation, \cite{Vaidman2003,BeigiKonig2011}, complexity measures like $T$-count or the garden-hose model \cite{Speelman2016,BuhrmanFehrSchaffnerSpeelman2013}, holography and the AdS/CFT correspondence \cite{May2019,MayPeningtonSorce2020,May2022}, information-theoretic cryptography \cite{AllerstorferEtAl2024,AsadiEtAl2025CDS}, communication complexity \cite{GirishMayParhamYuen2026},  secret sharing \cite{AnanthEtAl2024}, Hamiltonian simulation \cite{ApelEtAl2024}, or Banach-space geometry \cite{JungeKubickiPalazuelosPerezGarcia2022}. A recent review gives a unified account of these connections \cite{May2026Review}.

Despite this broad range of connections, the main theoretical question remains open: how must the resources of the adversaries scale as a function of the resources of the honest prover? There are several possible resource measures. Some previous lower bounds compare only quantum resources, treating classical communication or computation as free \cite{BluhmChristandlSpeelman2022}, in a spirit related to the quantum bounded storage model \cite{DamgardFehrSalvailSchaffner2008}. Here we deliberately use a more skeptical measure. We compare the total dimension of all systems used by the adversaries, including entanglement registers and communicated systems, against $n$, the dimension of each honest input system in the NLQC task. In particular, classical and quantum registers are both counted through their dimension.

With respect to this total-resource measure, known generic attacks require resources exponential in $n$, whereas the best unconditional lower bounds remain much smaller; in particular, they are sublinear\footnote{The bound is linear in the number of qubits but with a prefactor strictly smaller than one.}  in $n$ \cite{BluhmChristandlSpeelman2022,TomamichelFehrKaniewskiWehner2013}. Closing this gap is the central motivation for this work.

If one could exhibit QPV protocols for which every attack requires exponential total resources, then QPV would be secure in a strong resource-bounded information-theoretic sense. Such lower bounds would also have consequences beyond position verification. For instance, by Theorem~14 of \cite{ApelEtAl2024}, they would rule out the existence of very good universal simulators, in the sense of Definition~9 of the same work.

Moreover, in view of the recent connections between NLQC and classical cryptographic primitives such as conditional disclosure of secrets and private simultaneous messages \cite{AllerstorferEtAl2024,AsadiEtAl2025CDS}, they could open a route to stronger lower bounds on the resources required for those primitives.

In this paper, building on the vision and techniques started in \cite{JungeKubickiPalazuelosPerezGarcia2022}, we will use Banach space geometry tools and ideas to provide exponential lower bounds in the resources of the adversaries under a low-degree assumption on their attacks.

More concretely, we consider the following family of games $G_\varepsilon$, indexed by a bit string $\varepsilon\in \{\pm 1\}^{n^2}$, which is known to the participants, Alice and Bob, which are the coordinated adversaries aiming to break the QPV protocol. 
The referee distributes the following initial state 
\begin{equation}\label{eq:psi}
    \ket{\psi}=\frac1n\sum_{i,j=1}^n \ket{ij}_R\ket{i}_A\ket{j}_B,
\end{equation}
keeps $R$, gives $A$ to Alice and $B$ to Bob, and asks Alice and Bob to produce
\begin{equation}\label{eq:psi-eps}
    \ket{\psi_\varepsilon}=\frac1n\sum_{i,j=1}^n \varepsilon_{ij}\ket{ij}_R\ket{i}_A\ket{j}_B.
\end{equation}
Equivalently, the honest operation is the diagonal sign unitary $U_\varepsilon\ket{i}\ket{j}=\varepsilon_{ij}\ket{i}\ket{j}$.  The referee accepts using the two-outcome measurement
\begin{equation}
   \bigl\{\ketbra{\psi_\varepsilon}{\psi_\varepsilon},\,\mathbbm{1}-\ketbra{\psi_\varepsilon}{\psi_\varepsilon}\bigr\}.
\end{equation}
The allowed strategies (after purification) by Alice and Bob are as follows (dependence on $\varepsilon$ is explicitly specified):

They share an initial entangled state $\ket{\phi_\varepsilon}$ living in systems $a,e;b,f$, where Alice holds systems $ae$ and Bob holds systems $bf$. They do a first unitary operation $V_\varepsilon$ acting on $Aae$ and $W_\varepsilon$ acting on $Bbf$. Then Alice and Bob interchange the systems $e,f$ (we call $T$ the map from $AaeBbf$ to $AafBbe$ which acts as identity on $A,a,B,b$ and swaps systems $e,f$) and they apply a second unitary operation $\tilde{V}_\varepsilon$ acting on $Aaf$ and $\tilde{W}_\varepsilon$ acting on $Bbe$. A strategy $\mathcal{S}$ is then completely specified by unitaries $V_\varepsilon, W_\varepsilon, \tilde{V}_\varepsilon, \tilde{W}_\varepsilon$ and state $\ket{\phi_\varepsilon}$. Here, $\dim(A)=\dim(B)=n$ and we can assume $\dim(a)=\dim(b)=\dim(e)=\dim(f)=k$.

The main result of the paper is the following (see Theorem \ref{thm:low-degree-direct} for the formal statement):

\vspace{0.38cm}

\noindent\fbox{%
\parbox{\textwidth}{
{\bf Main result:} For every $0<\delta<1/4$, if the first-round strategy operator given by $(V_\varepsilon \otimes W_\varepsilon)|\phi_\varepsilon\rangle$ is a polynomial in $\varepsilon$ of degree at most $n^{1/4-\delta}$, then there exists an $\varepsilon$ for which the adversaries Alice and Bob require auxiliary dimension  $k=\exp\left(
\Omega\left(
\frac{n^{2\delta}}{(\log n)^2}
\right)
\right)$ to win the game $G_\varepsilon$ with constant probability of success.
}}
\vspace{0.1cm}

Note that, since the function describing the first-round strategy is defined on $\{-1,1\}^{n^2}$, it is always a polynomial of degree at most $n^2$.

Note also that the probability of success, called the value of the game $\omega(G_\varepsilon)$, is precisely the entanglement fidelity between the target unitary $U_\varepsilon$ and the actual quantum channel implemented by Alice and Bob in their strategy, i.e. the fidelity between the corresponding Choi states. It is therefore a natural way of measuring the accuracy of the NLQC implementation of $U_\varepsilon$ done by Alice and Bob, and it will be the one used in this paper. If one prefers to measure the accuracy with other distance, like the diamond norm, one can simply transfer the results using the known bounds between them.

The low-degree result can be applied to several natural families of attacks in which the dependence on the public sign string is generated in a controlled way, in particular, by applying the target unitary $U_\varepsilon$ itself. This is the case for port-based teleportation attacks, where all the first-round dependence on $\varepsilon$ is contained in a factor
of the form $U_\varepsilon^{\otimes N}$. Our main result shows that all attacks with this property  have exponential lower bounds to the auxiliary dimension (Remark \ref{rem:port-based}). A similar situation holds  if one wants to study the query complexity of protocols for NLQC. In that direction, a no-free-lunch theorem is stated in Remark \ref{rem:query-compl}.

Finally, we will state a conjecture on the scaling of the type 2 constant of a particular Banach space, whose affirmative solution would imply the existence of a game $G_\varepsilon$ which, without any extra hypothesis, requires exponential auxiliary dimension to be won with constant probability of success (see Theorem \ref{thm:conjectural}).

\subsection{Relation to prior work}

Let us clarify how our formulation compares with previous constructions in which large bit strings also appear in the description of the game. The closest predecessor is \cite{JungeKubickiPalazuelosPerezGarcia2022}, which introduced the Banach-space route to position-based cryptography. In the game $G_{\rm Rad}$ studied there, the sign string is part of the protocol input: the verifiers sample $\{-1,1\}^{n^2}$ during the execution of the protocol, and this classical information is sent to one side of the prover. The lower bounds are then expressed in terms of regularity parameters measuring how the cheating strategy varies with this input.

The $f$-routing line of work is also closely related \cite{BuhrmanFehrSchaffnerSpeelman2013,BluhmChristandlSpeelman2022}. In the usual $f$-routing setting, both parties receive classical strings and these strings determine the routing, measurement, or processing of a small quantum system. By contrast, in $G_{\rm Rad}$ only one side receives the sign string, but in both constructions the relevant bit string is still part of the input supplied during the execution of the task. Consequently, the known lower bounds become genuinely non-trivial only in a regime where classical information is treated as a free resource.

The present formulation removes this feature. Here $\varepsilon$ is fixed and public: it labels the game $G_\varepsilon$ rather than being supplied as a long classical challenge during the execution of the protocol. Thus our aim is not to lower bound the cost of receiving, processing, or communicating this classical string. Rather, we study the family $G_\varepsilon$ and ask whether at least one fixed diagonal unitary in this family is intrinsically costly to implement as a non-local quantum computation when all resources are counted. The conjectural lower bound is obtained by averaging the optimized game values and proving that at least one member is hard for all strategies at the given budget. For the low-degree result, we instead bound the average success of each admissible family of first-round operators, optimized over all possible second rounds.

As an important comment, \cite{BluhmChristandlSpeelman2022} requires a generic Boolean function for the argument to work. In subsequent recent papers \cite{AsadiCulfMay2025,AsadiCleveCulfMay2025}, the authors have managed to obtain similar bounds for particular easy functions, such as the inner product.

In this sense, to the best of our knowledge, this is the first paper in which exponential lower bounds (though under extra hypothesis) are provided if all resources (classical and quantum) are taken into consideration. 

\subsection{Structure of the paper}

In Section \ref{sec:basics} we give the necessary background on Banach space theory needed in the paper, in particular tensor norms, polynomial inequalities, complex interpolation and type constants.  In Section \ref{sec:game} we explain the game $G_\varepsilon$, and detail the most general strategies the participants can use. In Section \ref{sec:conjectural} we introduce the Banach spaces which capture the set of allowed strategies. We state and prove there Theorem \ref{thm:conjectural}, an exponential lower bound on the resources of the players based on a conjecture about the type 2 constant of a particular Banach space. Section \ref{sec:low-degree} contains the main result of the paper, Theorem \ref{thm:low-degree-direct}, and its proof, together with some applications in Remarks \ref{rem:port-based} and \ref{rem:query-compl}.

\section{Facts and notation from Banach-space theory}\label{sec:basics}
We will require the use of Banach space theory tools. For completeness, we will sketch here the notions, results and notations which will be required later. We refer to \cite{JungeKubickiPalazuelosPerezGarcia2022} and references therein for more detailed information. All spaces will be complex and finite dimensional. Constants denoted by $C$ are independent of the dimensions; their permitted dependence on other parameters will be stated.

\subsection*{Basic definitions}
For a quantum system $A$, $\H_A$ denotes its Hilbert space. We fix orthonormal bases in all Hilbert spaces. We write $S_p^{A,B}$ for the linear maps from $\H_B$ to $\H_A$, endowed with the Schatten norm of order $p$; thus the row (output) index comes first and the column (input) index comes second. For $p=1,2,\infty$,
\[
\|T\|_{S_1}=\Tr|T|,\qquad
\|T\|_{S_2}=\bigl(\Tr(T^\dagger T)\bigr)^{1/2},\qquad
\|T\|_{S_\infty}=\sup_{\|\xi\|=1}\|T\xi\|,
\]
where $|T|=(T^\dagger T)^{1/2}$. Here $T^\dagger$ is the Hilbert-space adjoint, whereas $T^T$ denotes the transpose in the fixed bases. A superscript $*$ on a Banach space denotes its space of continuous \emph{linear} functionals, that is, its dual space. We use the bilinear trace pairing
\begin{equation}\label{eq:trace-duality}
\langle x,y\rangle_{\rm tr}=\Tr(y^Tx),
\qquad x\in S_1^{A,B},\quad y\in S_\infty^{A,B}.
\end{equation}
It gives $(S_1^{A,B})^*=S_\infty^{A,B}$ isometrically. In particular,
$\|x\|_1=\sup_{\|y\|_\infty\leq1}|\Tr(y^Tx)|$.
Tensor duality is the bilinear extension of the pairings on the individual factors.

For normed spaces $X,Y$, the operator norm of $T:X\to Y$ is
\begin{equation}
\|T:X\to Y\|=\max_{x\in B_X}\|T(x)\|_Y,
\end{equation}
where $B_X$ is the closed unit ball. If $E\subset X$ is the range of a projection of norm at most one, we call $E$ a $1$-complemented subspace. A tensor norm induces the same norm on tensor products of such subspaces: apply its metric mapping property both to the inclusions and to the projections. Inequalities between two norms on the same vector space are reversed on the dual space.

\subsection*{Tensor norms}

A tensor norm $\alpha$ assigns to every pair of normed spaces $X,Y$
a norm $\|\cdot\|_\alpha$ on the algebraic tensor product $X\otimes Y$
such that

\begin{enumerate}
\item[(i)] $\|x\otimes y\|_{X\otimes_\alpha Y}
=
\|x\|_X\|y\|_Y$ for every $x\in X$ and $y\in Y$.

\item[(ii)] \emph{(metric mapping property)} for every pair of bounded
linear operators $T_i:X_i\to Y_i$, $i=1,2$,
\[
\|T_1\otimes T_2:
X_1\otimes_\alpha X_2
\longrightarrow
Y_1\otimes_\alpha Y_2\|
\leq
\|T_1\|\,\|T_2\|.
\]
\end{enumerate}

In particular, $\|f\otimes g\|_{(X\otimes_\alpha Y)^*}
=
\|f\|_{X^*}\|g\|_{Y^*}$ for every $f\in X^*$ and $g\in Y^*$.
The smallest and largest such tensor norms are, respectively, the injective
and projective tensor norms. For $u\in X\otimes Y$, we define
\begin{align}
\|u\|_\varepsilon
&=
\sup_{f\in B_{X^*},\,g\in B_{Y^*}}
|(f\otimes g)(u)|,
\label{eq:injective-definition}\\
\|u\|_\pi
&=
\inf\left\{
\sum_{s=1}^r\|x_s\|_X\|y_s\|_Y:
u=\sum_{s=1}^r x_s\otimes y_s,\ r\in\mathbb N
\right\}.
\label{eq:projective-definition}
\end{align}
Consequently, $\varepsilon\leq\alpha\leq\pi$.
In finite dimensions, $(X\otimes_\varepsilon Y)^*$ is canonically
and isometrically identified with $X^*\otimes_\pi Y^*$, and
$(X\otimes_\pi Y)^*$ with $X^*\otimes_\varepsilon Y^*$, under the
canonical tensor pairing
$\langle f\otimes g,x\otimes y\rangle=f(x)g(y)$.
We refer to \cite{DefantFloret1993} and
\cite[Section~2.2]{JungeKubickiPalazuelosPerezGarcia2022} for background.
\subsection*{Polynomials}
Our degree estimate will reduce to a one-variable derivative bound. A real polynomial with values in a Banach space $\mathcal B$ is a finite sum $q(t)=\sum_{r=0}^N b_rt^r$, where $b_r\in\mathcal B$ and the variable $t$ is real. We recall the standard, complex valued, Bernstein's inequality and deduce its vector-valued extension.

\begin{thm}[Bernstein's inequality]\label{thm:bernstein-complex}
Let $p$ be a complex valued real polynomial of degree at most $N$. Then, for every  $x\in(-1,1)$,
\begin{equation}\label{eq:bernstein-complex}
|p'(x)|
\leq
\frac{N}{\sqrt{1-x^2}}
\sup_{t\in[-1,1]}|p(t)|.
\end{equation}
See \cite[Chapter~5, Theorem~5.1.7]{BorweinErdelyi1995}.
\end{thm}
\begin{prop}[Bernstein's inequality, Banach-valued case]\label{prop:bernstein-banach}
For a polynomial $q$ of degree at most $N$ with values in a Banach space $\mathcal B$,
\begin{equation}\label{eq:bernstein-banach2}
\|q'(x)\|_{\mathcal B}\leq\frac{N}{\sqrt{1-x^2}}
\sup_{t\in[-1,1]}\|q(t)\|_{\mathcal B},\qquad -1<x<1.
\end{equation}
In particular, $\|q'(0)\|_{\mathcal B}\leq N\sup_{[-1,1]}\|q\|_{\mathcal B}$.
\end{prop}
\begin{proof}

For $\varphi\in B_{\mathcal B^*}$, the scalar polynomial $\varphi\circ q$ has degree at most $N$. The inequality (\ref{eq:bernstein-complex}) gives
\[
|\varphi(q'(x))|\leq\frac{N}{\sqrt{1-x^2}}
\sup_{t\in[-1,1]}|\varphi(q(t))|
\leq\frac{N}{\sqrt{1-x^2}}
\sup_{t\in[-1,1]}\|q(t)\|_{\mathcal B}.
\]
Taking the supremum over $\varphi\in B_{\mathcal B^*}$ and using Hahn--Banach proves \eqref{eq:bernstein-banach2}.
\end{proof}

\subsection*{Interpolation}

Properties of interpolation spaces will allow us to obtain estimates that are useful for our purposes. Here we restrict ourselves to complex interpolation of finite-dimensional Banach spaces. We will not recall the full definition of the complex interpolation method, and instead focus on the properties that we need; see \cite{BerghLofstrom1976} for a complete treatment.

In our finite-dimensional setting, the interpolation space $(X_0,X_1)_\theta$, $0<\theta<1$, can always be constructed. More generally, for arbitrary Banach spaces, one says that the couple $(X_0,X_1)$ is compatible\footnote{A couple $(X_0,X_1)$ is compatible if both spaces are continuously embedded in a common Hausdorff topological vector space.} when the complex interpolation construction is well defined. For the sake of concreteness, throughout this work we consider the case in which $X_0$, $X_1$ and $(X_0,X_1)_\theta$ are algebraically the same finite-dimensional vector space endowed with different norms.

The complex interpolation method, which assigns to every compatible couple $(X_0,X_1)$ the space $[X_0,X_1]_\theta$, is an exact interpolation functor of exponent $\theta$. This means that it satisfies the following property.

\begin{thm}[{\cite[Theorem~4.1.2]{BerghLofstrom1976}}]\label{Interp_thm} Let $(X_0,X_1)$ and $(Y_0,Y_1)$ be compatible couples, and let $T:X_0+X_1\to Y_0+Y_1$ be linear and bounded from $X_i$ to $Y_i$, $i=0,1$. Then
\[
\|T:[X_0,X_1]_\theta\to[Y_0,Y_1]_\theta\|
\le
\|T:X_0\to Y_0\|^{1-\theta}
\|T:X_1\to Y_1\|^\theta.
\]
Here $\|\cdot\|$ denotes the usual operator norm.
\end{thm}

We will only need the case $\theta=\frac12$. By applying this functorial property to tensor products, the interpolation of two tensor norms gives another tensor norm: the metric mapping property follows by interpolation, while the crossnorm property is preserved since the endpoint tensor norms agree on elementary tensors. In particular, we will consider extensively the $\frac12$-interpolation tensor norm between the two extremal tensor norms $\varepsilon$ and $\pi$. We denote this norm by $(\varepsilon,\pi)_{\frac12}$. We will use the following Hilbertian identification:
$
\mathcal H_A\otimes_{(\varepsilon,\pi)_{\frac12}}\mathcal H_B
=
\mathcal H_{AB}$ isometrically. In particular,
$
S_2^{A,B}\otimes_{(\varepsilon,\pi)_{\frac12}}S_2^{C,D}
=
S_2^{AC,BD}
$
isometrically.

\subsection*{Type constants}
Given a normed space $X$, we define its type $2$ constant $T_2(X)$ as the smallest constant $C$ such that, for every finite sequence of vectors   $x_1,\ldots,x_N\in X$,
\[
\left(\E_\epsilon\left\|\sum_{i=1}^N\epsilon_i x_i\right\|_X^2\right)^{1/2}
\le C\left(\sum_{i=1}^N\|x_i\|_X^2\right)^{1/2}.
\]
Here the expectation is taken with respect to the uniform measure on $\{-1,1\}^N$.

Replacing the $L_2$ mean by the $L_1$ mean changes this constant by universal factors, by the Kahane--Khintchine inequality; see \cite[Comment 2.12]{JungeKubickiPalazuelosPerezGarcia2022}. Restricting to sequences of length at most $N$ defines $T_2^{(N)}(X)$.

We will need the following non-trivial estimate, obtained from
\cite[Proposition~2.18 and the estimate following equation~(13)]
{JungeKubickiPalazuelosPerezGarcia2022}: if $d_A\leq d_B$, then
\begin{equation}\label{eq3}
T_2^{(d_A^2)}
\left(
S_1^{A,B}\otimes_{(\varepsilon,\pi)_{1/2}}S_1^{A,B}
\right)
\leq
C d_A^{3/4}\log(2d_A)\sqrt{\log(2d_Ad_B)}.
\end{equation}

\section{The game}\label{sec:game}
Fix $\varepsilon\in\{\pm1\}^{n^2}$, which is known to the participants. The referee prepares the states
\[
\ket{\psi}=\frac1n\sum_{i,j=1}^n\ket{ij}_R\ket{i}_A\ket{j}_B, \,\,\text{ and }\,\, 
\ket{\psi_\varepsilon}=\frac1n\sum_{i,j=1}^n\varepsilon_{ij}\ket{ij}_R\ket{i}_A\ket{j}_B,
\]
keeps $R$, and sends $A$ to Alice and $B$ to Bob. Here $\dim A=\dim B=n$ and $\dim R=n^2$. The final measurement is
$\{\Pi_\varepsilon,I_{RAB}-\Pi_\varepsilon\}$, where
$\Pi_\varepsilon=\ketbra{\psi_\varepsilon}{\psi_\varepsilon}$.
The value of the game will be the optimal probability of success among all allowed strategies which, according to the NLQC restriction (see Figure \ref{fig:causal-structure}), will consist of two rounds of local quantum operations (quantum channels), using a pre-shared quantum state, with an intermediate round of simultaneous communication. 

\subsection*{The strategies}
Initially Alice and Bob share a state, that we can assume to be pure $\rho_{A'B'}=|\tau\rangle\langle\tau|$, $\tau\in \H_{A'B'}$. Their first-round channels and their second-round channels have the following form. We will always use a tilde to denote the second round. 
\begin{equation}\label{eq:channel-spaces}
\begin{aligned}
\mathcal N_A &: \mathcal L(\H_{AA'})\longrightarrow\mathcal L(\H_{a_0e}),
&\mathcal M_B &: \mathcal L(\H_{BB'})\longrightarrow\mathcal L(\H_{b_0f}),\\
\widetilde{\mathcal N_A} &: \mathcal L(\H_{a_0f})\longrightarrow\mathcal L(\H_A),
&\widetilde{\mathcal M_B} &: \mathcal L(\H_{b_0e})\longrightarrow\mathcal L(\H_B).
\end{aligned}
\end{equation}

Thus $a_0,b_0$ are kept locally, while $e$ and $f$ are exchanged simultaneously. We temporarily use $a_0,b_0$ because the final local registers $a,b$ will also contain dilation spaces. Every channel is completely positive and trace preserving. 

The finite-dimensional Stinespring theorem gives, for a channel $\mathcal N:\mathcal L(H)\to\mathcal L(K)$, an auxiliary Hilbert space $H_0$ and an isometry $V:H\to K\otimes H_0$ such that (see \cite[Corollary 2.27, in particular items 5--6]{Watrous2018})
\begin{equation}\label{eq:stinespring}
\mathcal N(\rho)=\Tr_{H_0}(V\rho V^\dagger),
\qquad V^\dagger V=I_H,
\qquad \dim H_0\le(\dim H)(\dim K).
\end{equation}

Applying this representation to the four channels in \eqref{eq:channel-spaces} gives the isometries:
\begin{equation}\label{eq:raw-isometries}
\begin{aligned}
V_A^0 &: \H_{AA'}\longrightarrow\H_{a_0eH_A},
&W_B^0 &: \H_{BB'}\longrightarrow\H_{b_0fH_B},\\
\widetilde{V}_A^0 &: \H_{a_0f}\longrightarrow\H_{AK_A},
&\widetilde{W}_B^0 &: \H_{b_0e}\longrightarrow\H_{BK_B}.
\end{aligned}
\end{equation}

Let $\rho_0=\ketbra{\psi}{\psi}\otimes\ketbra{\tau}{\tau}$. With the tensor factors ordered according to the indicated registers, the first round, exchange, and second round give
\begin{align}
\rho_1&=(I_R\otimes V_A^0\otimes W_B^0)\rho_0
                    (I_R\otimes V_A^0\otimes W_B^0)^\dagger,\nonumber\\
\rho_2&=(I_R\otimes I_{H_A}\otimes I_{H_B}\otimes\Sigma)\rho_1(I_R\otimes I_{H_A}\otimes I_{H_B}\otimes\Sigma)^\dagger,\nonumber\\
\rho_3&=(I_R\otimes \widetilde{V}_A^0\otimes I_{H_A}\otimes \widetilde{W}_B^0\otimes I_{H_B})\rho_2
        (I_R\otimes \widetilde{V}_A^0\otimes I_{H_A}\otimes \widetilde{W}_B^0\otimes I_{H_B})^\dagger.
\label{eq:channel-evolution}
\end{align}
Here $\Sigma:\H_{a_0e}\otimes\H_{b_0f}\to\H_{a_0f}\otimes\H_{b_0e}$ only exchanges $e,f$, leaving $a_0,b_0$ unchanged; the identities on the dilation spaces $H_A,H_B$ are written separately. Canonical permutations are used to put each map next to its input registers; they do not change any norm. The physical output and its acceptance probability are
\begin{equation}\label{eq:channel-success}
\rho_{\rm out}=\Tr_{K_AH_AK_BH_B}\rho_3,
\qquad
p_\varepsilon(\mathcal S)=\Tr\big[(\Pi_\varepsilon\otimes I_{K_AH_AK_BH_B})\rho_3\big]
=\Tr(\Pi_\varepsilon\rho_{\rm out}).
\end{equation}
Here $\mathcal S$ denotes the fixed channel strategy. Thus, the referee is never acted on by the adversaries; the identities on $R$ and on the unmeasured output registers are explicit.

By adjoining fixed pure ancillas, embedding the input and output spaces of each isometry into a common local workspace, and completing orthonormal bases, we can define the first and second round by using unitaries 
\[
V:\H_{Aae}\to\H_{Aae},\qquad
W:\H_{Bbf}\to\H_{Bbf},
\]
\[
\widetilde V:\H_{Aaf}\to\H_{Aaf},\qquad
\widetilde W:\H_{Bbe}\to\H_{Bbe},
\] and a new state \[
\ket{\phi}\in\H_{ae}\otimes\H_{bf},
\]which reproduce the original previous actions without introducing
exponentially larger registers.

For the norm calculation we fix the purified and padded registers. By embedding smaller registers and extending the unitaries, we may assume
\begin{equation}\label{eq:dimension-convention}
\dim A=\dim B=n,\qquad \dim a=\dim b=\dim e=\dim f=k.
\end{equation}
Thus $\dim(ae)=\dim(bf)=\dim(af)=\dim(be)=k^2$. The parameter $n$ is fixed by the game, whereas $k$ depends on the strategy. If $K$ bounds the original auxiliary dimensions, \eqref{eq:stinespring} bounds all dilation dimensions polynomially in $n,K$. Absorbing the environments into local workspaces and performing the above embeddings therefore gives $k\leq\operatorname{poly}(n,K)$. Consequently, exponential lower bounds in the padded convention imply exponential lower bounds on the original channel registers, up to a polynomial change of dimension.

The exchange is the unitary permutation $T : \H_{Aae}\otimes\H_{Bbf}
\longrightarrow
\H_{Aaf}\otimes\H_{Bbe}$, defined by 
\begin{equation}\label{eq:exchange-full}
T(\ket{i}_A\ket{\alpha}_a\ket{u}_e
  \ket{j}_B\ket{\beta}_b\ket{v}_f)
=
\ket{i}_A\ket{\alpha}_a\ket{v}_f
\ket{j}_B\ket{\beta}_b\ket{u}_e.
\end{equation}

Thus, $T$ acts as the identity on $A,a,B,b$ and exchanges only the communicated registers $e$ and $f$. Write
\[
E_{\rm in}=aebf,\qquad E_{\rm out}=afbe,
\]
for the residual registers before and after the exchange, and put
\[
\mathcal U=(\widetilde V\otimes\widetilde W)
T(V\otimes W).
\]
The full final vector is
\[
(I_R\otimes\mathcal U)(\ket{\psi}\otimes\ket{\phi}),
\]
and the success probability of this fixed strategy is
\begin{equation}\label{eq:projected-vector}
p_\varepsilon(\mathcal S)=
\left\|
(\bra{\psi_\varepsilon}_{RAB}\otimes I_{E_{\rm out}})
(I_R\otimes\mathcal U)
(\ket{\psi}_{RAB}\otimes\ket{\phi}_{E_{\rm in}})
\right\|_{\H_{E_{\rm out}}}^2.
\end{equation}

Now, let us write the expression in Equation (\ref{eq:projected-vector}) in a way that will be very useful later. To this end, define the input and output blocks by
\begin{equation}\label{eq:block-spaces}
\begin{aligned}
V_i&=V(\ket{i}_A\otimes I_{ae})
    :\H_{ae}\to\H_{Aae},
&
W_j&=W(\ket{j}_B\otimes I_{bf})
    :\H_{bf}\to\H_{Bbf},
\\
\widetilde V_i
&=(\bra{i}_A\otimes I_{af})\widetilde V
    :\H_{Aaf}\to\H_{af},
&
\widetilde W_j
&=(\bra{j}_B\otimes I_{be})\widetilde W
    :\H_{Bbe}\to\H_{be}.
\end{aligned}
\end{equation}
Since these blocks come from common unitaries, their constraints must be
kept jointly. In particular,
\begin{equation}\label{eq:block-constraints}
\begin{gathered}
V_i^\dagger V_r=\delta_{ir}I_{ae},
\qquad
W_j^\dagger W_s=\delta_{js}I_{bf},
\\
\sum_i\widetilde V_i^\dagger\widetilde V_i=I_{Aaf},
\qquad
\sum_j\widetilde W_j^\dagger\widetilde W_j=I_{Bbe}.
\end{gathered}
\end{equation}
Thus the output blocks are contractions, but they are not in general
isometries individually. We always optimize over the full unitaries,
rather than over independently chosen blocks.

Expanding the two referee vectors in \eqref{eq:projected-vector}, the
scalar products in $R$ give
$\langle rs|ij\rangle=\delta_{ri}\delta_{sj}$. Hence the remaining
vector is
\begin{equation}\label{eq:amplitude-vector}
\frac1{n^2}\sum_{i,j=1}^n
\varepsilon_{ij}
(\widetilde V_i\otimes\widetilde W_j)
T(V_i\otimes W_j)\ket{\phi}
=
\frac1{n^2}\Gamma_\varepsilon\ket{\phi},
\end{equation}
where
\[
\Gamma_\varepsilon:\H_{E_{\rm in}}\to\H_{E_{\rm out}}
\]
denotes the operator appearing in the sum. Since each of the two
normalized referee vectors contributes a factor $1/n$, squaring the
amplitude gives the factor $1/n^4$. Let $\omega_k(G_\varepsilon)$ denote the supremum of $p_\varepsilon(\mathcal S)$ over all strategies in the fixed padded spaces. Optimizing first over the shared
state therefore yields
\begin{equation}\label{eq:game-operator-value}
\omega_k(G_\varepsilon)
=
\frac1{n^4}
\sup_{V,W,\widetilde V,\widetilde W}
\|\Gamma_\varepsilon\|_
{S_\infty^{E_{\rm out},E_{\rm in}}}^{\,2}.
\end{equation}
Equivalently, by Hilbert-space duality,
\begin{equation}\label{eq1}
\omega_k(G_\varepsilon)
=
\frac1{n^4}
\sup
\left|
\sum_{i,j=1}^n
\varepsilon_{ij}
\bra{\widetilde\phi}
(\widetilde V_i\otimes\widetilde W_j)
T(V_i\otimes W_j)
\ket{\phi}
\right|^2,
\end{equation}
where the four maps are unitaries on the spaces above and
$\|\phi\|=\|\widetilde\phi\|=1$. The vector
$\widetilde\phi$ only dualizes the norm of the unmeasured final vector;
it is not an additional entangled state available to the adversaries.

The Hilbert norm on $E_{\rm out}$ in
\eqref{eq:projected-vector} is therefore only an intermediate norm. Our
goal in the next section is to identify a Banach-space norm that captures the
\emph{whole optimization} in \eqref{eq1}.

\section{The relevant Banach spaces, and a conjectural exponential lower bound}\label{sec:conjectural}
We use throughout the row--column convention and bilinear trace duality of \eqref{eq:trace-duality}. We keep the notation $\sigma$. For the first round define
\begin{equation}\label{eq:first-round-operator}
P_{V,W,\phi}=\sum_{i,j=1}^n(V_i\otimes W_j)\ket{\phi}\bra{ij}
:\H_{AB}\longrightarrow\H_{AaeBbf}.
\end{equation}
In other words, $P_{V,W,\phi}=(V\otimes W)(I_{AB}\otimes\ket{\phi})$, with the factors in local order. On $S_1^{Aae,A}\otimes S_1^{Bbf,B}$ put
\begin{equation}\label{eq:sigma-first}
\|x\|_\sigma=\sup_{V,W,\phi}|\Tr(x^TP_{V,W,\phi})|,
\end{equation}
where $V,W$ are contractions on their full square spaces and $\|\phi\|=1$. Replacing the contractions by unitaries does not change this supremum since every square contraction is a convex combination of unitaries.

For the second round define
\begin{equation}\label{eq:second-round-operator}
Q_{\widetilde V,\widetilde W,\widetilde\phi}
=\sum_{i,j=1}^n\ket{ij}\bra{\widetilde\phi}
(\widetilde V_i\otimes\widetilde W_j)
:\H_{AafBbe}\longrightarrow\H_{AB},
\end{equation}
and, on $S_1^{A,Aaf}\otimes S_1^{B,Bbe}$, put
\begin{equation}\label{eq:sigma-second}
\|y\|_\sigma=\sup_{\widetilde V,\widetilde W,\widetilde\phi}
|\Tr(Q_{\widetilde V,\widetilde W,\widetilde\phi}y^T)|.
\end{equation}
Again the full square maps are contractions and $\|\widetilde\phi\|=1$, or equivalently the full maps may be required to be unitaries. Transposition turns $Q$ into a first-round operator with transposed matrices and the coordinate-conjugate vector $\overline{\widetilde\phi}$. Thus \eqref{eq:sigma-second} is the corresponding first-round norm transported by transposition, not by a bilinear use of the adjoint.

Every linear map $V:\H_{Aae}\to\H_{Aae}$ is naturally identified with a linear map $R_V:S_1^{Aae,A}\to\H_{ae}$. In the fixed orthonormal basis $(\ket{\eta})$ of $\H_{ae}$, this identification is
\[
R_V(x)=\sum_\eta
\Tr\!\left[x^TV(I_A\otimes\ket{\eta})\right]\ket{\eta}.
\]
Define $R_W$ analogously. Expanding the bilinear trace pairing and taking the supremum over $\phi$ gives
\begin{equation}\label{eq:sigma-first_alt}
\|x\|_\sigma=\sup_{V,W}\|(R_V\otimes R_W)(x)\|_{\H_{ae}\otimes_2\H_{bf}},
\end{equation}
where the supremum is over the full contractions $V,W$, and $\H_{ae}\otimes_2\H_{bf}=\H_{aebf}$.  It follows then that the previous formula for $\|x\|_\sigma$ defines a norm and the same happens for the expression for $\|y\|_\sigma$. 

We have therefore obtained the spaces
\begin{equation}\label{eq:X-spaces}
\boxed{\begin{aligned}
X&=S_1^{Aae,A}\otimes_\sigma S_1^{Bbf,B},\\
\widetilde X&=S_1^{A,Aaf}\otimes_\sigma S_1^{B,Bbe}.
\end{aligned}}
\end{equation}
The order of the indices is not optional: $P$ maps $AB$ \emph{into} the first-round output, whereas $Q$ maps the second-round input \emph{into} $AB$.

Moreover, by definition we have
\begin{equation}\label{eq:polar-balls}
B_{X^*}=\overline{\operatorname{aconv}}\{P_{V,W,\phi}\},
\qquad
B_{\widetilde X^*}=\overline{\operatorname{aconv}}\{Q_{\widetilde V,\widetilde W,\widetilde\phi}\},
\end{equation}
where $\operatorname{aconv}$ is the absolutely convex hull. 

Let us consider the element 
\begin{align*}
\Phi_\varepsilon=\frac1{n^2}\sum_{i,j=1}^n\varepsilon_{ij}T\otimes \ketbra{ij}{ij}_{AB}\in \tilde{X}\otimes X,
\end{align*}where here $T$ denotes the tensor associated to the map defined in (\ref{eq:exchange-full}). More precisely, we can write 
\begin{equation}\label{T_epsilon_tensor}
\Phi_\varepsilon=\frac1{n^2}\sum_{i,j=1}^n\sum_{AaeBbf}\varepsilon_{ij}\, (|i\rangle\langle Aaf|\otimes |j\rangle\langle Bbe|)\otimes (|Aae\rangle\langle i|\otimes |Bbf\rangle\langle j|).
\end{equation}

The following result expresses the value of the game as the injective tensor norm of this tensor.
\begin{prop}\label{prop:exact-value}
For the fixed normal-form dimensions in \eqref{eq:dimension-convention},
\begin{equation}\label{eq:value-as-norm}
\omega_k(G_\varepsilon)=\|\Phi_\varepsilon\|_{\widetilde X\otimes_\epsilon X}^{\,2}.
\end{equation}
\end{prop}
\begin{proof}
By the definition of the injective norm, a convexity argument and \eqref{eq:polar-balls},
\begin{align*}
\|\Phi_\varepsilon\|_{\widetilde X\otimes_\epsilon  X}^{\,2}&=\sup_{P_{V,W,\phi}, Q_{\widetilde V,\widetilde W,\widetilde\phi}}
\Big|\Tr\Big((Q_{\widetilde V,\widetilde W,\widetilde\phi}\otimes P_{V,W,\phi})\Phi_\varepsilon^T\Big)\Big|\\&=\frac1{n^4}\sup_{P_{V,W,\phi}, Q_{\widetilde V,\widetilde W,\widetilde\phi}}\Big|\sum_{i,j=1}^n
\varepsilon_{ij}
\langle \widetilde\phi|(\widetilde V_i\otimes\widetilde W_j)
T(V_i\otimes W_j)\ket{\phi}\Big|^2,
\end{align*}which corresponds to $\omega_k(G_\varepsilon)$ according to Eq. (\ref{eq:game-operator-value}).

\end{proof}

Since the norm $\sigma$ is tricky to work with, we will upper-bound it with a more tractable one.
\begin{lemma}\label{lem:bound-sigma-interpolation}
On each of the two oriented tensor products in \eqref{eq:X-spaces},
$\sigma\leq(\epsilon,\pi)_{1/2}$.
\end{lemma}
\begin{proof}
We only prove the statement for $X$ since the proof for $\widetilde X$ is completely analogous.

As we have seen in Eq. (\ref{eq:sigma-first_alt}), we can write 
\begin{align*}
\|x\|_\sigma=\sup_{V,W}\|(R_V\otimes R_W)(x)\|_{\H_{ae}\otimes_2\H_{bf}},
\end{align*}where the supremum is taken over all contractions $V:\H_{Aae}\to \H_{Aae}$ and $W:\H_{Bbf}\to \H_{Bbf}$ and $\H_{ae}\otimes_2\H_{bf}=\H_{aebf}$. Now, for the contractions $V$ ad $W$, it is easy to deduce that $\|R_V:S_1^{Aae,A}\to \H_{ae}\|\leq 1$ and $\|R_W:S_1^{Bbf,B}\to \H_{bf}\|\leq 1$. Hence, the metric mapping property of the $\epsilon$ and $\pi$ norms implies \begin{align*}&\|R_V\otimes R_W:S_1^{Aae,A}\otimes_{\pi}S_1^{Aae,A}\to \H_{ae}\otimes_{\pi} \H_{bf}\|\leq 1,\\&\|R_V\otimes R_W:S_1^{Aae,A}\otimes_{\epsilon}S_1^{Aae,A}\to \H_{ae}\otimes_{\epsilon} \H_{bf}\|\leq 1.
\end{align*} Then, the result follows from Theorem \ref{Interp_thm} and the well known fact $(\H_{ae}\otimes_{\epsilon} \H_{bf}, \H_{ae}\otimes_{\pi} \H_{bf})_{1/2}=\H_{ae}\otimes_2 \H_{bf}$.
\end{proof}

\begin{rem}\label{rem:gamma2}
The same argument also works with $(\epsilon,\gamma_2^*)_{1/2}$. Here, for a tensor $u\in U\otimes V$, $\gamma_2(u)$ is the infimum of $\|a\|\|b\|$ over factorizations of its associated map $U^*\to V$ as $b a$ through a Hilbert space, and $\gamma_2^*$ is its dual tensor norm under the bilinear tensor pairing; see \cite{DefantFloret1993}. On two Hilbert factors, $\gamma_2$ is the operator norm and $\gamma_2^*$ is the trace norm. Therefore their relevant interpolation endpoint is again Hilbertian and
\[
\sigma\le(\epsilon,\gamma_2^*)_{1/2}\le(\epsilon,\pi)_{1/2}.
\]
This observation concerns the domination of $\sigma$. It does not identify the two interpolated spaces in general, nor does it transfer \eqref{eq3} to the smaller norm. We retain $(\epsilon,\pi)_{1/2}$ in the estimates below.
\end{rem}

Consequently,
\begin{equation}\label{eq:Z-bound}
\omega_k(G_\varepsilon)\le\|\Phi_\varepsilon\|_Z^2,
\end{equation}
where $Z=\widetilde Y\otimes_\epsilon Y$ and 
\begin{equation}\label{eq:Y-spaces}
\boxed{\begin{aligned}
Y&=S_1^{Aae,A}\otimes_{(\epsilon,\pi)_{1/2}}S_1^{Bbf,B},\\
\widetilde Y&=S_1^{A,Aaf}\otimes_{(\epsilon,\pi)_{1/2}}S_1^{B,Bbe}
\end{aligned}}
\end{equation}
Under the padding convention \eqref{eq:dimension-convention}, transposition identifies $Y$ and $\widetilde Y$ isometrically. Without such a dimension convention this identification need not hold, since $\dim(ae)$ and $\dim(af)$ may differ.

The following tensor estimate will be used below.  Let us denote $F_{i,j}=T\otimes \ketbra{ij}{ij}$ for every $i,j=1,\ldots, n$.
\begin{lemma}\label{lem:corners}
For all $i,j=1,\ldots,n$ we have $\|F_{ij}\|_Z=1$.
\end{lemma}
\begin{proof}
Following the notation in Eq. (\ref{T_epsilon_tensor}), we need to study the norm of the element
$$F_{ij}=\sum_{AaeBbf} (|i\rangle\langle Aaf|\otimes |j\rangle\langle Bbe|)\otimes (|Aae\rangle\langle i|\otimes |Bbf\rangle\langle j|)$$ in $$\Big(S_1^{A,Aaf}\otimes_{(\epsilon,\pi)_{1/2}}S_1^{B,Bbe}\Big)\otimes_\epsilon \Big(S_1^{Aae,A}\otimes_{(\epsilon,\pi)_{1/2}}S_1^{Bbf,B}\Big).$$

Now, since $i$ and $j$ are fixed, we see that we are actually working on rows or columns of the $S_1$ spaces, which are actually 1-complemented spaces. Hence, we have 
\begin{align*}\|F_{ij}\|_Z&=\Big\|\sum_{AaeBbf} (e_{Aaf}\otimes e_{Bbe})\otimes (e_{Aae}\otimes e_{Bbf})\Big\|_{(\H_{Aaf}\otimes_{(\epsilon,\pi)_{1/2}}\H_{Bbe})\otimes_{\epsilon}(\H_{Aae}\otimes_{(\epsilon,\pi)_{1/2}}\H_{Bbf})}\\&=\Big\|\sum_{AaeBbf} (e_{Aaf}\otimes e_{Bbe})\otimes (e_{Aae}\otimes e_{Bbf})\Big\|_{(\H_{Aaf}\otimes_2\H_{Bbe})\otimes_{\epsilon}(\H_{Aae}\otimes_2\H_{Bbf})}
\end{align*}

Now, this norm is exactly the operator norm $$\|T:\H_{AafBbe}\to \H_{AaeBbf}\|,$$which equals one because $T$ is an isometry.
\end{proof}

We propose the following conjecture.
\begin{conj}\label{conj-weak}
There are constants $0\le\beta<1$, $q>0$, and $C$, independent of $n,k$, such that
\[
T_2(Z)\le Cn^\beta[\log(\dim Z)]^q.
\]
\end{conj}
It is a notorious difficult problem to estimate the type 2 constant of the injective tensor product of Banach spaces. The most common and powerful approach to give lower bounds is to consider the so-called volume ratio of the dual space. The estimate one obtains in this case, which can be deduced similarly as in \cite[Theorem 5.4]{JungeKubickiPalazuelosPerezGarcia2022}, is compatible with the conjecture, with $\beta=3/4$ and $q=1$.

\begin{thm}[Conjectural lower bound]\label{thm:conjectural}
Assume Conjecture~\ref{conj-weak}. For every $0<\delta<(1-\beta)/q$ and every fixed $c>0$, all sufficiently large $n$ admit a sign string $\varepsilon$ for which every strategy with success probability at least $c$ requires padded local dimension $k>\exp(n^\delta)$. The corresponding lower bound on the original channel registers is $\exp(\Omega(n^\delta))$.
\end{thm}
\begin{proof}
For fixed $k$, \eqref{eq:Z-bound}, type $2$, and Lemma~\ref{lem:corners} give
\begin{equation}\label{eq:avg-bound}
\E_\varepsilon\omega_k(G_\varepsilon)
\le\E_\varepsilon\left\|\frac1{n^2}\sum_{i,j}\varepsilon_{ij}F_{ij}\right\|_Z^2
\le\frac{T_2^{(n^2)}(Z)^2}{n^2}
\le C^2 n^{-2(1-\beta)}[\log(\dim Z)]^{2q}.
\end{equation}
Here, $\dim Z=n^8k^8$. Choose $k=\lfloor\exp(n^\delta)\rfloor$. The last expression tends to zero because $q\delta<1-\beta$. Hence some $\varepsilon$ has $\omega_k(G_\varepsilon)<c$. Since $\omega_k$ already optimizes over \emph{all} strategies at that budget, this gives the stated quantifiers. 
Finally, the polynomial relation between $k$ and the original
dimension parameter $K$ established above implies the corresponding
exponential lower bound
$
K=\exp(\Omega(n^\delta)).
$
\end{proof}
In particular, $q=1$ and $\beta=3/4$ give any $\delta<1/4$ by this argument. A different exponent for $k$ would require a different, explicitly stated logarithmic estimate.

In the next section we obtain a direct resource lower bound from low-degree dependence of the full first-round operator on the public sign string. 

\section{Low-degree polynomial strategies}\label{sec:low-degree}

Let us fix a family of first-round strategies
\[
S^1=\{V^\varepsilon,W^\varepsilon,|\phi_\varepsilon\rangle\}_\varepsilon
\]
on the common padded registers fixed above. For this fixed first round, set
\[
\Psi_\varepsilon^{S^1}
=
\frac{1}{n^2}\sum_{i,j=1}^n
\varepsilon_{ij}\,
T\left(V_i^\varepsilon\otimes W_j^\varepsilon\right)
|\phi_\varepsilon\rangle\langle ij|.
\]
In order to regard this tensor as an element of $\widetilde X$, we consider its transpose,
$
\Psi_P(\varepsilon):=\left(\Psi_\varepsilon^{S^1}\right)^t\in\widetilde X.
$
Reasoning as in Proposition \ref{prop:exact-value}, by the definition of $\widetilde{X}$, one can readily see that $\|\Psi_P(\varepsilon)\|_{\widetilde X}^2$
is precisely the  optimal success probability one can get for the given fixed first round, optimized over all admissible
second rounds.
Let
$
P(\varepsilon)
=
\sum_{i,j=1}^n(V_i^\varepsilon\otimes W_j^\varepsilon)
\ket{\phi_\varepsilon}\bra{ij}
$
be the physical first-round family, viewed as a function
\[
P:\{\pm1\}^{n^2}\longrightarrow
X^*=
\left(S_1^{Aae,A}\otimes_\sigma S_1^{Bbf,B}\right)^*.
\]
By definition of $X^*$,
$
\|P(\varepsilon)\|_{X^*}\leq 1
$
for every $\varepsilon$. All registers are embedded into the same fixed
padded spaces for the whole family. The polynomial estimates below will
also apply to arbitrary $X^*$-valued functions satisfying this norm bound,
whether or not they arise from physical first-round strategies.

Set $m=n^2$ and identify $(i,j)\in[n]^2$ with $s\in[m]$. Let

\begin{equation}\label{eq:F-contractions}
F_s(R)
:=
E_sRT
\end{equation}

Thus, we have 
\begin{equation}\label{eq:Psi-correct}
\Psi_P(\varepsilon)
=
\langle\Id_{\widetilde X}\otimes P(\varepsilon),\Phi_\varepsilon\rangle
=
\frac1m\sum_{s=1}^m
\varepsilon_s F_s(P(\varepsilon)^T)
\end{equation}

We set
\[
a_{n,k}
:=
\frac1n T_2^{(m)}(Y)
=
\frac1n T_2^{(m)}(\widetilde Y).
\]
Applying \eqref{eq3}, after transposition when needed, to the rectangular dimensions $n$ and $nk^2$ gives
\begin{equation}\label{eq:a-type-bound}
a_{n,k}\leq C\frac{\log(2n)\sqrt{\log(2nk)}}{n^{1/4}}.
\end{equation}
Every function on the Boolean cube has a unique multilinear Fourier--Walsh expansion of degree at most $m$.

We now assume that $P$ has degree
at most $d$, with $\sup_\varepsilon\|P(\varepsilon)\|_{X^*}\leq1$. As we have mentioned before, we must work with the transpose object, for which we will use the notation $P^T(\epsilon):=P(\epsilon)^T$. We consider the Fourier--Walsh expansion of $P^T$.
\[
P^T(\varepsilon)
=
\sum_{\substack{S\subset[m]\\ |S|\leq d}}
\widehat P(S)\chi_S(\varepsilon),
\qquad
\chi_S(\varepsilon)
=
\prod_{r\in S}\varepsilon_r.
\]

The auxiliary polynomial introduced below is designed to
recover $\Psi_P$ by differentiation. The reason for the
particular factor
\(
\theta(1-\theta)^r
\)
is that its derivative at the origin is equal to one,
independently of $r$. Random restrictions will produce exactly
these factors for the different Fourier levels of $P$.

\begin{lemma}
\label{lem:H-properties}
Define, for $\theta\in\mathbb R$,
\[
H(\theta)(\varepsilon)
=
\frac1m
\sum_{\substack{S\subset[m]\\ |S|\leq d}}
\sum_{s=1}^m
\theta(1-\theta)^{|S\setminus\{s\}|}
\chi_{S\triangle\{s\}}(\varepsilon)
F_s(\widehat P(S)).
\]
Then $H$ is a polynomial with values in
$
\mathcal B
=
L_2(\{-1,1\}^m;\widetilde X)
$
and satisfies
\[
H(0)=0,
\qquad
H'(0)=\Psi_P,
\qquad
\deg H\leq d+1.
\]
\end{lemma}

\begin{proof}
The first assertion is immediate. Moreover, for every integer $r\geq0$,
\[
\left.
\frac{d}{d\theta}
\bigl[\theta(1-\theta)^r\bigr]
\right|_{\theta=0}
=1.
\]
It follows that
\begin{align*}
H'(0)(\varepsilon)
&=
\frac1m
\sum_{S,s}
\chi_{S\triangle\{s\}}(\varepsilon)
F_s(\widehat P(S))
\\
&=
\frac1m
\sum_s
\varepsilon_s
F_s\left(
\sum_S\widehat P(S)\chi_S(\varepsilon)
\right)
\\
&=
\frac1m
\sum_s
\varepsilon_sF_s(P^T(\varepsilon))
\\
&=
\Psi_P(\varepsilon),
\end{align*}
where we have used
$
\chi_{S\triangle\{s\}}(\varepsilon)
=
\varepsilon_s\chi_S(\varepsilon).
$ 
Finally,
\[
\deg_\theta
\left(
\theta(1-\theta)^{|S\setminus\{s\}|}
\right)
=
1+|S\setminus\{s\}|
\leq
|S|+1
\leq
d+1.
\]
Therefore
$
\deg H\leq d+1.
$
\end{proof}

We next show that the polynomial in
Lemma~\ref{lem:H-properties} appears naturally from random
restrictions and, more importantly, satisfies a useful norm
estimate.

\begin{prop}
\label{prop:H-random-restriction}
For a fixed $J\subset[m]$, define the conditional expectation
\[
(\mathcal E_JP^T)(\varepsilon)
=
2^{-|J|}
\sum_{\delta_J\in\{-1,1\}^J}
P^T(\varepsilon_{J^c},\delta_J)
\]
and
\[
A_J(\varepsilon)
=
\frac1m
\sum_{s\in J}
\varepsilon_s
F_s\bigl((\mathcal E_JP^T)(\varepsilon)\bigr),
\]
\[
B_J(\varepsilon)
=
\mathcal E_J
\left[
\frac1m
\sum_{s\in J}
\varepsilon_sF_s(P^T(\varepsilon))
\right].
\]
For $0\leq \theta\leq 1$, let $\mathbb P_\theta$ be the
probability measure on $2^{[m]}$ associated with selecting independently each $s\in [m]$ with probability $\theta$. That is,
\[
\mathbb P_\theta(J)
=
\theta^{|J|}(1-\theta)^{m-|J|},
\qquad
J\subseteq[m].
\]
Then
\[
H(\theta)
=
\mathbb E_{J\sim\mathbb P_\theta}
\left[
A_J+B_J
\right],
\]
where $H$ is the polynomial of
Lemma~\ref{lem:H-properties}. Moreover,
\[
\|H(\theta)\|_{\mathcal B}
\leq
2a_{n,k}\sqrt{\theta}.
\]

\end{prop}

\begin{proof}
We first estimate $A_J$ and $B_J$ for a fixed set $J$.

Fix $\varepsilon_{J^c}$ and write
$
R=(\mathcal E_JP^T)(\varepsilon).
$
Then $R$ is independent of the signs
$(\varepsilon_s)_{s\in J}$. Hence
\begin{align*}
\left(
\mathbb E_{\varepsilon_J}
\|A_J(\varepsilon)\|_{\widetilde X}^2
\right)^{1/2}
&\leq
\frac1m
\left(
\mathbb E_{\varepsilon_J}
\left\|
\sum_{s\in J}
\varepsilon_sF_s(R)
\right\|_{\widetilde Y}^2
\right)^{1/2}
\\
&\leq
\frac{T_2^{(m)}(\widetilde Y)}{m}
\left(
\sum_{s\in J}
\|F_s(R)\|_{\widetilde Y}^2
\right)^{1/2}
\\
&\leq
\frac{T_2^{(m)}(\widetilde Y)}{m}
\sqrt{|J|}.
\end{align*}

Notice that the type-$2$ estimate is applicable here precisely because
$(\mathcal E_JP^T)(\varepsilon)$ depends only on the coordinates
$\varepsilon_{J^c}$. Thus, once $\varepsilon_{J^c}$ is fixed, the vectors
$
F_s\bigl((\mathcal E_JP^T)(\varepsilon)\bigr)
$
are fixed with respect to the Rademacher variables
$(\varepsilon_s)_{s\in J}$ over which the type inequality is applied. On the other hand, we have used the inequality, $$\|F_s(R)\|_{\widetilde Y}=\|F_s((\mathcal E_JP^T)(\varepsilon))\|_{\widetilde Y}=\|\mathcal E_J[F_s(P^T)(\varepsilon)]\|_{\widetilde Y}\leq \|F_s(P^T)(\varepsilon)\|_{\widetilde Y}\leq 1$$for every $s$ and $\epsilon$. Indeed, if $s=(i,j)\in [n^2]$, then 
\begin{align}\label{Eq_norm_s_fix}
\|F_s(P^T)(\varepsilon)\|_{\widetilde Y}=\||ij\rangle\langle \overline{\phi}|(V_i^T\otimes W_j^T)T\|_{\widetilde Y}=\|\langle \overline{\phi}|(V_i^T\otimes W_j^T)T\|_{\H_{AafBbe}}\leq 1.
\end{align}

Averaging also over $\varepsilon_{J^c}$ gives
$
\|A_J\|_{\mathcal B}
\leq
a_{n,k}\sqrt{\frac{|J|}{m}}.
$
We now estimate $B_J$. Recall that
\[
B_J(\varepsilon_{J^c})
=
\mathbb E_{\varepsilon_J}
\left[
\frac1m
\sum_{s\in J}
\varepsilon_s F_s(P^T(\varepsilon))
\right].
\]
Fix $\varepsilon_{J^c}$. By duality, 
\begin{align*}
\|B_J(\varepsilon_{J^c})\|_{\widetilde X}
&=
\sup_{Q\in B_{\widetilde X^*}}
\left|
\frac1m
\mathbb E_{\varepsilon_J}
\sum_{s\in J}
\varepsilon_s
\left\langle F_s(P^T(\varepsilon)),Q\right\rangle_{\rm tr}
\right|
\\
&=
\sup_{Q\in B_{\widetilde X^*}}
\left|
\frac1m
\mathbb E_{\varepsilon_J}
\left\langle
\sum_{s\in J}
\varepsilon_s TQ^T E_s,
P(\varepsilon)
\right\rangle_{\rm tr}
\right|.
\end{align*}

We obtain
\begin{align*}
\|B_J(\varepsilon_{J^c})\|_{\widetilde X}
&\leq
\frac1m
\sup_{Q\in B_{\widetilde X^*}}
\mathbb E_{\varepsilon_J}
\left\|
\sum_{s\in J}
\varepsilon_s TQ^TE_s
\right\|_Y
\\
&\leq
\frac1m
T_2^{(m)}(Y)
\sup_{Q\in B_{\widetilde X^*}}
\left(
\sum_{s\in J}
\|TQ^TE_s\|_Y^2
\right)^{1/2}.
\end{align*}

Using that $\|TQ^TE_s\|_Y\leq 1$ for every $s$, which can be proved as in Eq. (\ref{Eq_norm_s_fix}), we obtain 
\[
\|B_J(\varepsilon_{J^c})\|_{\widetilde X}
\leq
\frac{T_2^{(m)}(Y)}{m}\sqrt{|J|}
\leq
a_{n,k}\sqrt{\frac{|J|}{m}}.
\]
Consequently,
$
\|B_J\|_{\mathcal B}
\leq
a_{n,k}\sqrt{\frac{|J|}{m}}.
$

 We now treat separately the contributions of
$A_J$ and $B_J$. 

For $s\in[m]$ and $A\subseteq[m]$, define
\[ 
f_s(J)
=
\mathbf 1_{\{s\in J\}},
\qquad
g_A(J)
=
\mathbf 1_{\{A\cap J=\varnothing\}}.
\]
For a fixed $J$, these are simply deterministic quantities in
$\{0,1\}$.

We shall repeatedly use the following elementary identity:
\[
\mathbb E_{J\sim\mathbb P_\theta}
\bigl[f_s(J)g_A(J)\bigr]
=
\begin{cases}
0, & s\in A,\\[1mm]
\theta(1-\theta)^{|A|}, & s\notin A.
\end{cases}
\]
Indeed, if $s\in A$, the conditions
\[
s\in J
\qquad\text{and}\qquad
A\cap J=\varnothing
\]
are incompatible. If $s\notin A$, set
\[
R=[m]\setminus(A\cup\{s\}).
\]
Every set $J$ contributing to the expectation can then be
written uniquely as
\[
J=\{s\}\cup K,
\qquad
K\subseteq R.
\]
Hence
\begin{align*}
\mathbb E_{J\sim\mathbb P_\theta}
\bigl[f_s(J)g_A(J)\bigr]
&=
\sum_{K\subseteq R}
\mathbb P_\theta(\{s\}\cup K)
\\
&=
\sum_{K\subseteq R}
\theta^{1+|K|}
(1-\theta)^{m-1-|K|}
\\
&=
\theta(1-\theta)^{|A|}
\sum_{K\subseteq R}
\theta^{|K|}
(1-\theta)^{|R|-|K|}
\\
&=
\theta(1-\theta)^{|A|}.
\end{align*}

We now apply this identity to the two contributions.

\begin{itemize}

\item
\emph{Contribution of $A_J$.}
For fixed $J$,
$
\mathcal E_J\chi_S
=
g_S(J)\chi_S.
$
Therefore
\[
A_J(\varepsilon)
=
\frac1m
\sum_{s,S}
f_s(J)g_S(J)
\varepsilon_s\chi_S(\varepsilon)
F_s(\widehat P(S)).
\]
Averaging with respect to $\mathbb P_\theta$ gives
\[
\mathbb E_{J\sim\mathbb P_\theta}
A_J(\varepsilon)
=
\frac1m
\sum_{s,S}
\mathbb E_{J\sim\mathbb P_\theta}
\bigl[f_s(J)g_S(J)\bigr]
\varepsilon_s\chi_S(\varepsilon)
F_s(\widehat P(S)).
\]

We now use the identity above with $A=S$.
If $s\in S$, then
\[
\mathbb E_{J\sim\mathbb P_\theta}
\bigl[f_s(J)g_S(J)\bigr]
=
0.
\]
If $s\notin S$, then
\[
\mathbb E_{J\sim\mathbb P_\theta}
\bigl[f_s(J)g_S(J)\bigr]
=
\theta(1-\theta)^{|S|}.
\]
Since
$
\varepsilon_s\chi_S(\varepsilon)
=
\chi_{S\triangle\{s\}}(\varepsilon),
$
we obtain
\[
\mathbb E_{J\sim\mathbb P_\theta}
A_J(\varepsilon)
=
\frac1m
\sum_S
\sum_{s\notin S}
\theta(1-\theta)^{|S|}
\chi_{S\triangle\{s\}}(\varepsilon)
F_s(\widehat P(S)).
\]

\item
\emph{Contribution of $B_J$.}
Using
$
\varepsilon_s\chi_S(\varepsilon)
=
\chi_{S\triangle\{s\}}(\varepsilon),
$
we can write, for fixed $J$,
\[
B_J(\varepsilon)
=
\frac1m
\sum_{s,S}
f_s(J)
g_{S\triangle\{s\}}(J)
\chi_{S\triangle\{s\}}(\varepsilon)
F_s(\widehat P(S)).
\]
Hence
\[
\mathbb E_{J\sim\mathbb P_\theta}
B_J(\varepsilon)
=
\frac1m
\sum_{s,S}
\mathbb E_{J\sim\mathbb P_\theta}
\bigl[
f_s(J)g_{S\triangle\{s\}}(J)
\bigr]
\chi_{S\triangle\{s\}}(\varepsilon)
F_s(\widehat P(S)).
\]

We now apply the same identity with
$
A=S\triangle\{s\}.
$
If $s\notin S$, then
$
s\in S\triangle\{s\},
$
and therefore
\[
\mathbb E_{J\sim\mathbb P_\theta}
\bigl[
f_s(J)g_{S\triangle\{s\}}(J)
\bigr]
=
0.
\]
If $s\in S$, then
\[
S\triangle\{s\}
=
S\setminus\{s\},
\]
and hence
\[
\mathbb E_{J\sim\mathbb P_\theta}
\bigl[
f_s(J)g_{S\triangle\{s\}}(J)
\bigr]
=
\theta(1-\theta)^{|S|-1}.
\]
It follows that
\[
\mathbb E_{J\sim\mathbb P_\theta}
B_J(\varepsilon)
=
\frac1m
\sum_S
\sum_{s\in S}
\theta(1-\theta)^{|S|-1}
\chi_{S\triangle\{s\}}(\varepsilon)
F_s(\widehat P(S)).
\]

\end{itemize}

Combining the two contributions, we obtain
\[
\mathbb E_{J\sim\mathbb P_\theta}
\bigl[A_J+B_J\bigr](\varepsilon)
=
\frac1m
\sum_{S,s}
\theta(1-\theta)^{|S\setminus\{s\}|}
\chi_{S\triangle\{s\}}(\varepsilon)
F_s(\widehat P(S)).
\]
Finally, by Minkowski's inequality,
\begin{align*}
\|H(\theta)\|_{\mathcal B}
&\leq
\mathbb  E_{J\sim\mathbb P_\theta}
\left(
\|A_{ J}\|_{\mathcal B}
+
\|B_{J}\|_{\mathcal B}
\right)
\\
&\leq
2a_{n,k}
\mathbb E_{J\sim\mathbb P_\theta}
\sqrt{\frac{|J|}{m}}.
\end{align*}
Since the square root is concave,
\[
\mathbb E_{J\sim\mathbb P_\theta}
\sqrt{\frac{|J|}{m}}
\leq
\sqrt{
\frac{\mathbb E_{J\sim\mathbb P_\theta} J}{m}
}
=
\sqrt{\theta}.
\]
Thus
\begin{equation}
\|H(\theta)\|_{\mathcal B}
\leq
2a_{n,k}\sqrt{\theta}.
\end{equation}
\end{proof}

We can now obtain the low-degree estimate as an
application of Bernstein's inequality.

\begin{thm}
\label{thm:low-degree-direct}
Let \(P\) be the first-round family introduced above and assume that
$
\deg P\leq d.
$
Then
\[
\left(
\mathbb E_\varepsilon
\|\Psi_P(\varepsilon)\|_{\widetilde X}^2
\right)^{1/2}
\leq
4(d+1)a_{n,k}.
\]

If, in addition, the corresponding family of physical strategies
has average success probability at least a constant \(c_0>0\), then
\[
\log k
\geq
c\,
\frac{\sqrt n}
{(d+1)^2[\log(2n)]^2}
-
\log(2n),
\]
where \(c>0\) depends only on \(c_0\) and on universal constants.

In particular, for every fixed \(0<\delta<1/4\), if
$
d\leq n^{1/4-\delta},
$
then
\[
k
\geq
\exp\left(
\Omega\left(
\frac{n^{2\delta}}{(\log n)^2}
\right)
\right).
\]
\end{thm}
\begin{proof}
By Lemma~\ref{lem:H-properties}, $H(0)=0$. Hence there exists a
$\mathcal B$-valued polynomial $K$ of degree at most $d$ such
that
$
H(\theta)=\theta K(\theta).
$
Define, for $t\in[-1,1]$,
\[
q(t)
=
\begin{cases}
H(t^2)/t, & t\neq0,\\
0, & t=0.
\end{cases}
\]
Since
$
q(t)=tK(t^2),
$
$q$ is a genuine $\mathcal B$-valued polynomial and
$
\deg q\leq2d+1.
$
Moreover,
\[
q'(0)=K(0)=H'(0)=\Psi_P.
\]

By Proposition~\ref{prop:H-random-restriction},
\[
\|H(\theta)\|_{\mathcal B}
\leq2a_{n,k}\sqrt{\theta},
\qquad
0\leq\theta\leq1.
\]
Therefore, for $t\neq0$,
\[
\|q(t)\|_{\mathcal B}
=
\frac{\|H(t^2)\|_{\mathcal B}}{|t|}
\leq
2a_{n,k},
\]
and the same bound holds at $t=0$ by continuity.

Applying Proposition~\ref{prop:bernstein-banach} at $0$ to this $\mathcal B$-valued polynomial gives
\[
\|\Psi_P\|_{\mathcal B}
=\|q'(0)\|_{\mathcal B}
\leq(2d+1)\sup_{t\in[-1,1]}\|q(t)\|_{\mathcal B}
\leq2(2d+1)a_{n,k}
\leq4(d+1)a_{n,k}.
\]

If $P$ is physical and the corresponding strategies have
average success probability at least $c_0$, then
\[
c_0
\leq
\mathbb E_\varepsilon
\|\Psi_P(\varepsilon)\|_{\widetilde X}^2
\leq
16(d+1)^2a_{n,k}^2.
\]
Using
$
a_{n,k}
\leq
C\,
\frac{\log(2n)\sqrt{\log(2nk)}}{n^{1/4}},
$
we obtain
\[
c_0
\leq
C(d+1)^2
\frac{[\log(2n)]^2\log(2nk)}{\sqrt n}.
\]
Thus
\[
\log k
\geq
c\,
\frac{\sqrt n}
{(d+1)^2[\log(2n)]^2}
-
\log(2n).
\]

If $d\leq n^{1/4-\delta}$, then
$
(d+1)^2=O(n^{1/2-2\delta}),
$
and therefore
\[
\log k
\geq
\Omega\left(
\frac{n^{2\delta}}{(\log n)^2}
\right),
\]
which proves the claimed exponential lower bound. 
\end{proof}

We finish with two remarks about the implications of Theorem~\ref{thm:low-degree-direct}.

\begin{rem}\label{rem:port-based}
The state-of-the-art attack to quantum position verification schemes is based on port-based
teleportation \cite{BeigiKonig2011, Christandl-Asymptotic-2020}, and it requires exponential dimension in the adversaries. In the port-based teleportation attack, all the dependence on \(\varepsilon\)
in the first round comes from a term of the form
$
U_\varepsilon^{\otimes N},
$
while the remaining maps and the initial state are independent of
\(\varepsilon\). A consequence of Theorem \ref{thm:low-degree-direct} is that {\em any} attack with that property necessarily requires exponential dimension. The proof is very simple. For every fixed  \(0<\delta<1/4\), if
$
N\leq n^{1/4-\delta},
$
Theorem~\ref{thm:low-degree-direct} gives
\[
k\geq
\exp\left(
\Omega\left(
\frac{n^{2\delta}}{(\log n)^2}
\right)
\right).
\]

On the other hand, suppose that
$
N\geq n^{1/4-\delta}.
$
The dimension of the Hilbert space where $
U_\varepsilon^{\otimes N}
$ acts on is \(n^N\). Since, in our padded
convention,
$
\dim(aebf)=k^4,
$
we obtain
$
k^4\geq n^N,
$
and therefore
\[
k\geq n^{N/4}
\geq
\exp\left(
\frac14 n^{1/4-\delta}\log n
\right).
\]
Thus, in either regime, attacks of this form require auxiliary dimension
exponential in a positive power of \(n\). In particular, choosing
\(\delta=1/12\) gives the uniform bound
\[
k\geq
\exp\left(
\Omega\left(
\frac{n^{1/6}}{(\log n)^2}
\right)
\right).
\]
\end{rem}

\begin{rem}\label{rem:query-compl}
Theorem~\ref{thm:low-degree-direct} has also a related nice interpretation in terms of query complexity for NLQC. Consider protocols $\mathcal{P}$ that implement in NLQC any unkown bipartite unitary \(U_{AB}\), to which the parties are given only black-box access. Define the query complexity of protocol $\mathcal{P}$ as the minimum number $q(\mathcal{P})$ of black-box uses needed to implement any unknown unitary \(U_{AB}\) within a constant accuracy threshold. As in the rest of the text, let us assume $n=\dim(A)=\dim(B)$.

For instance, the query complexity of the port-based teleportation implementation of NLQC is $q(\mathcal{P}_{\rm port-based}) = O(n^4)$ \cite{Christandl-Asymptotic-2020}.

A corollary of Theorem~\ref{thm:low-degree-direct} is the following no-free-lunch result for NLQC: if there exists a protocol $\mathcal{P}$ for NLQC which requires only polynomial auxiliary dimension, then its query complexity must be $q(\mathcal{P})>n^{1/4-\delta}$.

The proof is simple. Since the unitary to be implemented \(U_\varepsilon\) is unknown, the only dependency on $\varepsilon$ in the protocol $\mathcal{P}$ (restricted to the subset of unitaries \(U_\varepsilon\)) is through the $q$ query uses of  \(U_\varepsilon\).  Then the dependence
on \(\varepsilon\) has polynomial degree at most \(q\). Therefore, for every
fixed \(0<\delta<1/4\), if
$
q\leq n^{1/4-\delta},
$
Theorem~\ref{thm:low-degree-direct} gives
\[
k\geq
\exp\left(
\Omega\left(
\frac{n^{2\delta}}{(\log n)^2}
\right)
\right).
\]
\end{rem}

\section*{On the use of AI}
AI has not been used at all in this paper. 

\section*{Acknowledgments}
We are grateful to Alex May and Marius Junge for very frutiful conversations during the preparatio of this work. 
This work has been funded by the Spanish Ministry of Science, Innovation and Universities MICIU/AEI/10.13039/501100011033 (CEX2023-001347-S, PID2023-146758NB- I00), Comunidad de Madrid (TEC-2024/COM-84-QUITEMAD-CM), Universidad Complutense de Madrid (FEI-EU-22-06),  and the Ministry for Digital Transformation and of Civil Service of the Spanish Government through the QUANTUM ENIA project call - Quantum Spain project, and by the European Union through the Recovery, Transformation and Resilience Plan - NextGenerationEU within the framework of the Digital Spain 2026 Agenda. 

This research was supported in part by Perimeter Institute for Theoretical Physics. Research at Perimeter Institute is supported by the Government of Canada through the Department of Innovation, Science, and Economic Development, and by the Province of Ontario through the Ministry of Colleges and Universities.

\end{document}